\documentclass[a4paper]{article}

\usepackage{fullpage,url,amsmath,amsfonts,amssymb,tedmath}

\newcommand{\mC}{\mathcal{C}}
\date{}

\title{MDS codes over finite fields}
\author{Ted Hurley\footnote{National University of Ireland
Galway. Ted.Hurley@NuiGalway.ie}}%, Donny Hurley\footnote{Institute of
\begin{document}
\maketitle
\begin{abstract}%\let\thefootnote\relax\footnote{
The mds (maximum distance separable) conjecture claims that  a nontrivial
linear mds $[n,k]$ code over the finite field $GF(q)$ %% , where $q$ is
%% a power of prime $p$,
satisfies  $n \leq (q + 1)$, except when $q$ is even and $k = 3$ or
$k = q- 1$ in which case it satisfies $n \leq (q + 2)$.

For given field $GF(q)$ and any given $k$, series of mds $[q+1,k]$ codes
 are constructed.   

Any $[n,3]$ mds or $[n,n-3]$ mds code over $GF(q)$ must satisfy $n\leq (q+1)$  for $q$ odd and $n\leq (q+2)$ for $q$ even. For even  $q$, mds $[q+2,3]$ and mds $[q+2, q-1]$ codes are
 constructed over $GF(q)$. 

The codes constructed have efficient encoding and  decoding algorithms.
%Proofs are given.  
\end{abstract}
\section{Introduction}

Background on coding theory and related material made be found in
\cite{Mac} or in \cite{blahut}. Now $GF(q)$ denotes the finite field of order $q$ and $q$  is necessarily a  power of a prime. % or in many of the writings in the area.  
 An $[n, k]$ linear code over  $GF(q)$ is a  linear code
$\mC$  of length $n$ and dimension $k$ over $GF(q)$. % Now $q$ must be a
% power of a prime $p$, $q = p^m$.   

The minimum distance $d$ of $\mC$ is
bounded by the Singleton bound $d \leq  (n + 1 - k)$.  If $d = (n + 1- k)$, then
the code $\mC$ is termed  a {\em maximum distance separable} (mds)
code. The mds codes are those with maximum 
error correcting capability for a given length and dimension. %% A class of these MDS codes is the one
%% of the so-called Reed-Solomon codes, which is of great importance in
%% modern industrial applications.
MacWilliams and
Sloane refer to  mds codes in their book \cite{Mac} as ``one of the most fascinating chapters in
all of coding theory'';  mds codes are equivalent to geometric objects
called {\em $n$-arcs} and combinatorial objects called {\em orthogonal arrays},
\cite{Mac},  and  are, quote,  ``at the heart of combinatorics and finite
geometries''.

The {\em mds conjecture} is due originally   %initially was %essentially
 to  Segre \cite{segre} from  1955.

{\bf mds conjecture:} If $\mC$ is a nontrivial
linear mds $[n,k]$ code over $GF(q)$, %%  where $q$ is
%% a power of prime $p$,
 then $n \leq (q + 1)$, except when $q$ is even and $k = 3$ or
$k = (q- 1)$ in which case $n \leq (q + 2)$.

There is a large literature focusing on this problem, for example see \cite{blok,ald,Thas}. %% Many papers focus on this
%% conjecture (see for instance, [1,4,25]).
Ball showed \cite{Ball} that the mds conjecture is true for prime fields.  %% Ball and Beule \cite{Beule}  showed 
%% that the conjecture is true  for $k \leq 2p- 2$. %% See also \cite{blok,Tha
%s} for further lists. %%
For a list of
when the conjecture is known to hold for q non-prime, see \cite{hirs,hirs1}

Here for any given finite field $GF(q)$ and any given $k$, series of mds
$[q+1,k]$ codes over $GF(q)$ are constructed.

Methods in \cite{hurley}
can be adopted to give efficient decoding algorithms; the complexity is
$\max\{O(\log n), t^2\}$ where $n$ is the length and $t$ is the
error-correcting capability which is $\floor{\frac{d-1}{2}}$ where $d$
is the distance.

For even  $q$, it is shown that any $[n,3]$ mds code and any $[n,n-3]$
mds code over $GF(q)$ must satisfy $n\leq (q+2)$ and for odd $q$ any
$[n,3]$ or $[n,n-3]$ mds code over $GF(q)$ must satisfy $n\leq (q+1)$.  
For even $q$, series of $[q+2,3]$ mds codes and $[q+2,q-1]$ mds codes over
$GF(q)$ are constructed. 

% For completeness a proof that any $[n,2]$ or mds $[n,n-2]$ over $GF(q)$
% must satisfy $n\leq q+1$ is given although the result is known. 

The mds codes constructed over prime fields are
maximum length for the field. The more general case is
dealt with separately.
\section{Basics}\label{basic}
A primitive $n^{th}$ root of unity in a field $\F$ is an element $\om$ such that $\om^n=1$ but $\om^i\neq 1$ for $1\leq i < n$. 

In a finite field $GF(q)$, a primitive $(q-1)$ root of unity always exists; see for example \cite{blahut,Mac} or any book on field theory.
Let $\om$ be a primitive element in $GF(q)$ so that $\om$ has order
$q-1=t$.  Then $S=\{1=\om^0, \om,\om^2,\ldots, \om^{t-1}\}$ are the
distinct elements of $GF(q)/\{0\}$,  $\om^t=1$ and $\om^{i} \neq 1$ for $1\leq i<t$.

See for example \cite{blahut} or \cite{Mac} for the following result. A
$k\ti n$ matrix $G$ is the generator matrix of an mds $[n,k]$ code if
and only if any $k\ti k$ submatrix of $G$ has non-zero determinant. Also
a $k\ti n$ matrix is a check matrix of a $[n,n-k]$ mds code if and only if
any $k\ti k$ submatrix has  non-zero determinant.

An   $[n,k]$ code is an mds code if and only if its dual is an $[n,n-k]$ mds code,  \cite{Mac,blahut}. 

Recall the mds codes constructed in \cite{hurley}.

% Recall the definition of a Fourier matrix definition over a field.
 A Fourier matrix is a special type of
Vandermonde matrix. Let $\om$ be a primitive $n^{th}$ root of unity in a
field $\F$; primitive here means that $\om^n=1$ and $\om^i \neq 1$ for $1\leq i < n$. The Fourier matrix $F_n$, relative to $\om$ and $\F$, is the $n\ti n$ matrix  

$$ F_n= \begin{pmatrix}1 & 1 & 1& \ldots & 1 \\ 1 & \om & \om^2 & \ldots &
	 \om^{n-1} \\ 
1 & \om^2 & \om^4 & \ldots & \om^{2(n-1)} \\ \vdots & \vdots & \vdots &
    \ldots & \vdots \\ 1 & \om^{n-1} & \om^{2(n-1)} & \ldots &
    \om^{(n-1)(n-1)} \end{pmatrix}$$

Simplifications can be made to some of the powers from $\om^n=1$. An $n^{th}$ root of unity can only
exist in a field provided the characteristic of the field does not
divide $n$ and in this case $n^{-1}$ exists. % in the field. 

Then  $$\begin{pmatrix}1 & 1 & 1& \ldots & 1 \\ 1 & \om & \om^2 & \ldots & 
	 \om^{n-1} \\ 
1 & \om^2 & \om^4 & \ldots & \om^{2(n-1)} \\ \vdots & \vdots & \vdots &
    \ldots & \vdots \\ 1 & \om^{n-1} & \om^{2(n-1)} & \ldots &
    \om^{(n-1)(n-1)} \end{pmatrix} \begin{pmatrix}1 & 1 & 1& \ldots & 1 \\ 1
				    & \om^{n-1} & \om^{2(n-1)} & \ldots
				    & \om^{(n-1)(n-1)}\\ 
1 & \om^{n-2} & \om^{2(n-2)} & \ldots &\om^{(n-1)(n-2)} \\ \vdots & \vdots & \vdots &
    \ldots & \vdots \\ 1 & \om & \om^{2} & \ldots &
    \om^{(n-1)} \end{pmatrix} = nI_n$$

The inverse of $F_n$ can be obtained from the above by multiplying
through by $n^{-1}$ when it exists. 

Recall the following from \cite{hurley}:

\begin{theorem}\label{genthmfour}{\cite{hurley}} 

(i) Let $F_n$ be a Fourier $n\ti n$ matrix over a field $\F$. % whose
%  characteristic does not divide $n$.
 Let $\mathcal{C}$ be a 
 code obtained by choosing in order $r$ rows of $F_n$ in
 arithmetic sequence with arithmetic 
difference $k$ satisfying   $\gcd(n,k) = 1$. Then
$\mathcal{C}$ is an mds $[n,r,n-r+1]$ code.

In particular this is true when $k=1$, that is, when the $r$ rows are
 chosen in succession.  

(ii) Let $\mathcal{C}$ be as in part (i). Then there exist explicit efficient encoding and decoding algorithms for $\mathcal{C}$.

\end{theorem}  

Thus series of mds codes are formed from rows of a Fourier matrix
using this {\em unit-derived} method developed initially in \cite{hur1}. %% In a finite field $GF(q)$, a primitive $(q-1)$ root of unity always exists; see for example \cite{blahut,Mac} or any book on field theory.

It is possible to choose rows which  `wrap over' and still get an
 mds code as long as the arithmetic difference $k$ between the rows
 is the same and satisfies
 $\gcd(k,n) = 1$  --  consider row $n+i$ the same as row $i$.

 Some of the methods of \cite{hurley} are generalizations of those of \cite{hurley1} but the papers are  independent.
 
In particular the following mds  codes are formed
from a finite field $GF(q)$ by using the primitive $(q-1)$ root. 

\begin{theorem}\label{1}(See \cite{hurley}) Let $GF(q)$ be a finite field and
 $\om$ a primitive $(q-1)$ root of unity in $GF(q)$. Form the Fourier
 $(q-1)\ti (q-1)$ matrix relative to $\om$: 
$$ F_{q-1}= \begin{pmatrix}1 & 1 & 1& \ldots & 1 \\ 1 & \om & \om^2 & \ldots &
	 \om^{q-2} \\ 
1 & \om^2 & \om^4 & \ldots & \om^{2(q-2)} \\ \vdots & \vdots & \vdots &
    \ldots & \vdots \\ 1 & \om^{q-2} & \om^{2(q-2)} & \ldots &
    \om^{(q-2)(q-2)} \end{pmatrix}$$
  Then choosing any $r$ rows of $F_{q-1}$ in arithmetic sequence with difference $k$ satisfying $\gcd(q-1,k)=1$ gives a generator matrix for an $[q-1,r]$ mds code. In particular taking consecutive rows gives an mds code.

  Further there exist explicit efficient encoding and decoding algorithms for the codes. 
\end{theorem}

%The sequence may wrap around and still get an mds code.

Consider cases where the first $r$ rows are chosen; other cases are similar. 

$A= \begin{pmatrix}1 & 1 & 1& \ldots & 1 \\ 1 & \om & \om^2 & \ldots &
	 \om^{q-2} \\ 
1 & \om^2 & \om^4 & \ldots & \om^{2(q-2)} \\ \vdots & \vdots & \vdots &
    \ldots & \vdots \\ 1 & \om^{r-1} & \om^{2(r-1)} & \ldots &
    \om^{(r-1)(q-2)} \end{pmatrix}$

This is a generator matrix for an $[q-1,r]$ mds code and  $A$ is an $r\ti (q-1)$ matrix.
%matrix and $n=q-1$ for the field $GF(q)$. 

Any $r\ti r$ submatrix of $A$ has non-zero determinant as $A$ generates an mds code. 

Now extend the length $(q-1)$ to $(q-1+2)=(q+1)$ as follows to  obtain an
  mds code $[q+1,r]$ code. 
  Extend $A$ by adding on two further $r\ti 1$ columns $v=(1,0, \ldots, 0)\T, w=(0,0,\ldots, 0,1)\T$ to obtain the $r\ti (q+1)$ matrix $B=(v,w,A)$. 

%% $B= \begin{pmatrix} 1 & 0 & a_{1,1} & \ldots & a_{1,q-1} \\ 0& 0 & a_{2,1} & \ldots & a_{2,q-1} \\ \vdots & \vdots &\vdots & \vdots \\ 0 & 0 & a_{r-1,1} & \ldots & a_{r-1,q-1} \\ 0&1 & a_{r,1} & \ldots & a_{r,q-1} \end{pmatrix}$.

%% In other words $B$. 

\begin{theorem}\label{main} The matrix $B$ generates an mds $[q+1,r]$ code.
\end{theorem}
\begin{proof} This is proved by showing that any $r\ti r$ submatrix of
 $B$ has non-zero determinant. 

If the $r\ti r$ submatrix is from $A$, only, then it has non-zero determinant
 since $A$ generates an mds $[q-1,r]$ code. 

Consider the case where the $r\ti r$ submatrix $P$ is formed by taking the
 first column of $B$ together with $(r-1)$ columns of $A$.

Then $P = (v, A')$ where $A'$ is an $r\ti (r-1)$ of $A$ with rows of
 $F_{q-1} $ in sequence. 

Evaluate  the determinant of $P$ by expanding via  the first column and get
 $\det(P) = \det(A^{''})$ where $A^{''}$ is an $(r-1)\ti (r-1)$ 
submatrix of $F_{q-1}$ with rows in sequence. Then $\det(A^{''}) \neq 0$ and
 so $\det(P) \neq 0$. 

Similarly the case where an $r\ti r$ submatrix formed by taking the
 second column of $B$ with $(r-1)$ columns  of $A$ can be shown to have
 non-zero determinant. 

Now form $Q=(v,w, A^{'''})$ where $A^{'''}$ consists of $(r-2)$ columns of
 $A$. Expand by first column and get that $\det(Q) = \det(w^{'},A^{iv})$
 where $w^{'}$ is $w$ with first zero omitted and $A^{iv}$ is $(r-1)\ti (r-2)$
 submatrix of $A$ with rows in sequence from $F_n$. Now expand by first
 column and get that $\det(Q) = \pm \det(A^v)$ where $A^v$ is from
 $F_n$ with rows in sequence and so $\det(A^v) \neq 0$. Hence $\det(Q)
 \neq 0$ as required.

\end{proof}

The result depends on the fact that any (square) $y\ti y$ submatrix of
$y$ rows {\em in
sequence} of the Fourier matrix have non-zero determinant
\cite{hurley1,hurley}; this requires the $\{v,w\}$ to have the forms 
given  -- with one starting with $1$ and the other ending with $1$ and all  other entries equal to zero.  

It is clear from Theorem \ref{genthmfour} that  the $A$ may be chosen  by taking $r$ 
rows of $F_{q-1}$ in arithmetic sequence with difference $k$ satisfying  $\gcd(k,q-1)=1$. 

Encoding and decoding is obtained by adapting the methods in \cite{hurley}
to the present situation.

More generally get the following result: 

\begin{theorem}\label{main1} Given the finite field $GF(q)$, form the Fourier
 $(q-1)\ti (q-1)$ matrix $F_{q-1}$ using a primitive $(q-1)$ element in
 $GF(q)$. Form the $r\ti (q-1)$ matrix $A$ by choosing $r$ rows of $F_{q-1}$ in
 arithmetic sequence with arithmetic difference $k$ satisfying
 $\gcd(k,q-1)=1$. Let $v=(1,0, \ldots, 0)\T, w=(0,0,\ldots, 0,1)\T$
 where these are of size $r\ti 1$. Let $B$ be the $r \ti (q+1)$ matrix
 obtained by adding $\{v,w\}$ as columns to $A$. Then $B$ is the generator
 matrix of an $[q+1,r]$ mds code.
\end{theorem}    

Moreover the methods in \cite{hurley} may be adopted to give efficient
encoding and decoding algorithms for the code generated by $B$. The
complexity of this is $\max\{O(n\log n), t^2\}$ where $t= \floor{\frac{d-1}{2}}$
with $d(=q-r +2)$ is the distance of the mds code $[q+1,r]$. 

\paragraph{Samples}

\begin{enumerate}

\item Let the field be $GF(3^2)$. The examples from this small field may be obtained directly   but are chosen to illustrate the general methods.  Let $\om$
      be a
  primitive $8^{th}$ root of unity in $GF(9)$.  Here $q=9, q+1=10$ by reference to Theorem \ref{main} or \ref{main1}. 
It is required to  construct  a $[10,r]$ mds code over $GF(9)$.

    Consider $r=4$ as an illustration; the construction for a general
      $r$ is similar.  

From the general construction above,  the following is an $[10,4]$ mds
      code. 

$\begin{pmatrix} 1 & 0 & 1 & 1 & 1 & \ldots &  1 \\ 0 & 0 & 1 & \om &\om^2 &
  \ldots &\om^7 \\ 0 &0 & 1& \om^2 & \om^4 & \ldots & \om^{14} \\ 0 &1 &
  1 & \om^3 & \om^6 & \ldots & \om^{21} \end{pmatrix}$

Note that $\om^8=1$ and some of the powers may be simplified. 

There are other possibilities by varying the matrix $A$ obtained from
      the Fourier matrix, Theorem \ref{main1}. %from the general method.
For example choose $2^{nd},5^{th}, 8^{th}$ rows (rows with arithmetic
      difference $3$ and  $\gcd(3,8)=1)$) of Fourier $F_{8}$ over
      $GF(9)$ using $\om$ to get 

$A=\begin{pmatrix} 1 & \om & \om^2 &
				     \ldots & \om^7 \\ 1 & \om^4 &
				     \om^{8} & \ldots & \om^{28} \\ 1 &
				     \om^7 &\om^{6} & \ldots & \om^{1}
				     \end{pmatrix}$.

Then add the two columns $(1,0,0)\T, (0,0,1)\T$ to the front of $A$ to get a
      $3\ti 10$ matrix $B$ which is then a generator matrix for a
      $[10,3]$ mds code over $GF(9)$.

Choosing $\{2^{nd},5^{th}, 8^{th}, 11^{th}= 3^{rd}\}$ rows, by
      wrapping,  to construct $A$ and then add the two columns as before to get $B$ which is then  a $[10,4]$ mds code over $GF(3^2)$.

There are many choices. 
\item Consider $GF(3^3)$. Here the $q=27$ and $q+1=28$ from general considerations.  Construct $[28,r]$ mds codes over $GF(27)$.

Let $\om$ be a primitive $26^{th}$ root of unity in $GF(27)$. Form the Fourier $F_{26\ti 26}$ matrix over $GF(27)$ using $\om$. 

Say $r=4$ for illustration; the more general $r$ is similar. 

Form $B=\begin{pmatrix} 1&0&1&1 & 1& \ldots & 1 \\ 0&0& 1&\om &\om^2& \ldots
       & \om^{25} \\ 0&0& 1& \om^2 &\om^4 & \ldots & \om^{50} \\ 0&1 &
       1&\om^3 &\om^6 & \ldots & \om^{75} \end{pmatrix} $

(Some of the powers may be simplified on noting $\om^{26}=1$.)
$B$ is formed using the first $4$ rows of $F_{26\ti 26}$ together with $v=(1,0,0,0)\T, w=(0,0,0,1)\T$.

Then $B$ is an $[28,4]$ mds code over $GF(27)$.

To get a $[28,24]$ mds code over $GF(27)$, take $B$ as the check matrix
      of a code. Alternatively take an $[26,24]$ mds code from the Fourier $26 \ti 26$ matrix and add on the two columns $\{v,w\}$ as before. %% This is the same as taking as generator matrix with
      %% two columns together with the 22 first rows of the Fourier matrix
      %% with $\om$ as primitive $26^{th}$ root of unity. 

      There are many other ways that the $A$ could be formed as noted previously and the $B$ is obtained by adding on the two extra columns, one beginning with $1$ and the other ending with $1$ and all other entries zero.  %by the general method.
    \item Consider the prime field $GF(257)$. Now the order of $(3 \mod 257)$ is $256$. Thus $\om = (3 \mod 257)$ is a primitive element in $GF(257)$. Here $q=257,  q+1 = 258$ from general considerations.
      Construct $[258, r]$ mds codes over $GF(257)$ as follows.

      Form the Fourier $256\ti 256$ matrix $F_{256}$ over $GF(257)$ using $\om = (3 \mod 256)$ as the primitive element. Choose $r$ rows of $F_{256}$ chosen is arithmetic sequence with difference $k$ satisfying $\gcd(k,256)= 1$ to form a $ r \ti 256$ matrix $A$. Now add the two columns $u= (1,0,0, \ldots, 0)\T, w = (0,0,0, \ldots, 0,1)\T$ of length $r$ to the front of $A$ to form a matrix $B$. Then $B$ generates an mds $[258,r]$ code.

      This code is the maximum length code that can be formed from $GF(257)$. Note also that the arithmetic for the codes is modular arithmetic performed in $GF(257)= \Z_{257}$ and with powers of $(3 \mod 257)$ only.  Note that efficient encoding and decoding algorithms exist of complexity $\max\{O(n\log n, t^2)$ where $n=256, t = \floor{\frac{d-1}{2}}$ where $d$ is the distance which equals $257-r$.
     
\end{enumerate}  
\section{ Even $q$ and dimension $3$}

Consider $GF(q)$ where $q$ is even.  %% and codes $[n,3], [n,n-3]$the dimension is $k=3$.

     %% In other words the differences $\mod n$ from top line to next line are the same.
\subsection{Sample}
Consider $GF(2^3)$ initially. Let $\om$ be a primitive $7^{th}$
root of unity in $GF(8)$. Let $A$ be the first three rows of the Fourier $7\ti 7$ matrix formed using $\om$. %Here we are able to extend a bit further.

Define 
$B=\begin{pmatrix}1&0&0 & 1 &1 &1& 1&1&1 & 1 \\ 0&1&0&1&\om & \om^2 &
  \om^3&\om^4&\om^5 & \om^6 \\ 0&0&1&1&\om^2 &\om^4 &\om^6&\om &\om^3 &\om^5\end{pmatrix}$

This is $B=(u,v,w,A)$ where $u=(1,0,0)\T, v=(0,1,0)\T, w=(0,0,1)\T$ and
$A$ is the first three rows of the Fourier $7\ti 7$ matrix formed using
$\om$ as the primitive $7^{th}$ root of unity. 

Show that $B$ is an mds $[10,3]$ code over $GF(8)$ as follows. 

Now any $3\ti 3 $ submatrix of $B$ involving columns of $A$, only, has non-zero 
determinant as $A$ generates an mds code, Theorem \ref{genthmfour}, \cite{hurley}.

If any of $\{u,v,w\}$ with two columns of $A$ are used to form a $3\ti 3 $ submatrix then for it to have a zero determinant it must be   that $A$ has a $2\ti 2$ submatrix with zero determinant. It 
may be verified that no $2\ti 2$ submatrix of $A$ has zero determinant
directly using Lemma \ref{prim} below; a direct proof of a more general result which includes this  is given  in  Proposition \ref{cor1}. %% which follows on from Lemma \ref{odd} and Corollary \ref{cor} below.

If a $3\ti 3$ matrix formed with two of $\{u,v,w\}$ together with a column of $A$ has a zero determinant
then this means that $A$ has  a zero element which it doesn't.

Thus $B$ is a generator of an mds $[10,3]$ code. 

To form an $[10,7]$ mds code we may take $B$ as the check matrix of a
$[10,3]$ code. %% In fact this code is the same as that obtained by taking
%% $u,v,w$ in conjunction with the first $7$ rows of the Fourier matrix
%% formed from $\om$.

There are many choices for a $3\ti 10$ matrix $A$ from the Fourier matrix
and then add on the three columns as noted to get a $B$ which is then a $[10,3]$ mds code.  
\subsection{General case }
  Consider $GF(2^n)$.  To construct $[2^n+2, 3]$ mds codes in $GF(2^n)$ consider the Fourier $F_{(2^n-1) \ti (2^n-1)}$ matrix formed using the primitive $(2^n-1)^{th}$ root of unity $\om$ in $GF(2^n)$. Take the first
  three rows, or any three rows in arithmetic sequence with difference
  $k$ satisfying $\gcd(k,2^n-1) = 1$, of this  Fourier matrix to form a $3 \ti (2^n-1)$ matrix $A$, which generates an mds code. Now form the matrix $B$ by 
  adding  the three columns of  $I_3$ to  the front (or anywhere indeed) of $A$.

  The proof that this $B$ 
  is an $[2^n+2,3]$ mds code then reduces to
  showing that this $3\ti (2^n-1)$ matrix $A$ from the Fourier matrix has no $2 \ti 2$  submatrix with determinant equal to zero. 

Using this $3\ti (2^n+2)$ matrix $B$ as a check matrix gives a $[2^n+2,2^n-1]$
mds code. %alternatively use the $(2^n-1)$ rows of the Fourier matrix.  

     \begin{lemma}\label{prim} Let $\om$ be a primitive $n^{th}$ root of unity. Then
       $\det( \begin{pmatrix} \om^i & \om^j \\ \om^k & \om^l \end{pmatrix}) = 0 $ if and only if $i-k \equiv j-l \mod n$.
     \end{lemma}
     \begin{proof} $\det( \begin{pmatrix} \om^i & \om^j \\ \om^k & \om^l \end{pmatrix}) = 0 $ if and only if $\om^i\om^l - \om^j\om^k = 0$ if and only if $\om^{i+l} = \om^{j+k}$ if and only if $i+l \equiv j+k \mod n$ if and only if $i-k \equiv j-l \mod n$. \end{proof}

\begin{lemma}\label{odd} Suppose in a field the order of $\om$ is $t$ where
  $t=2j+1$ is odd. 
Then the matrix 

$\begin{pmatrix} 1& 1 & 1 & \ldots & 1 \\ 1 & \om & \om^2 & \ldots &
 \om^{t-1} \\ 1 & \om^2 & \om^{4} & \ldots &
\om^{2t-2}\end{pmatrix}$ 

has no $2\ti 2$ submatrix with determinant
  equal to zero. 
\end{lemma}
\begin{proof}
Note that $\om^{2i-i} = \om^i$ and $\om^{2i}=\om^{2j}$ if and only if
 $\om^{2(i-j)} =1$ if and only if $\om^{i-j}=1$ as $\om$ has odd
 order; for $i,j < t$ this implies $i=j$. 

 It is clear that there are no $2\ti 2$ submatrices formed
 from the first row and either of the second or third rows with
 determinant zero as (i) $\{\om, \om^2, \ldots \om^{t-1}\}$ are distinct and
 (ii) $\{\om^2,\om^4, \ldots, \om^{2t-2}\}$ are distinct.

It remains to show that a $2\ti 2$ submatrix formed from second and
 third row cannot have determinant equal to zero. 

Work out the differences between the powers in the third row with those
 immediately above in the second row:
 $\{0, \om, \om^2, \ldots, \om^j, \om^{j+1}, \om^{j+2}, \ldots , \om^{2j}\}$.

 %as $\om^{-i} = \om^{t-i}$.
 
These are all different and so by Lemma \ref{prim} there is no $2\ti 2$ submatrix
 from second and third rows with determinant equal to zero.
\end{proof} 

\begin{corollary}\label{cor} Suppose $\om$ has odd order $t=2j+1$ in a field. 
Then if $A=\begin{pmatrix} 1& 1 & 1 & \ldots & 1 \\ 1 & \om & \om^2 & \ldots &
 \om^{t-1} \\ 1 & \om^2 & \om^{4}  & \ldots &
\om^{2t-2}\end{pmatrix}$ has no $3\ti 3$ submatrix with determinant zero
 then 
$B=  \begin{pmatrix}1&0&0& 1 & 1 &1& \ldots & 1 \\0&1&0 &1 & \om & \om^2 & \ldots &
 \om^{t-1} \\ 0&0&1& 1 & \om^2 & \om^{4} & \ldots &
\om^{2t-2}\end{pmatrix}$ has no $3\ti 3$ submatrix with determinant
 equal to zero. 

\end{corollary} 

\begin{proof} 
If the $3\ti 3$ submatrix is taken from $A$, only, then know it has non-zero
 determinant as $A$ generates an mds code. If the submatrix  involves one of the first three columns of $B$ then
 by expanding along this column its determinant is zero if and only if
 $A$ has a $2\ti 2$ submatrix whose determinant is zero; by Lemma
 \ref{odd} this cannot happen. If it involves two of the first three
 columns of $B$ then by expansion by these columns in order its
 determinant is non-zero as $A$ has no zero entry. If it involves all
 three of the first columns of $B$ then obviously the determinant  is $1$
 and is non-zero.  
\end{proof} 

\begin{proposition}\label{cor1} Form the Fourier matrix $F_{2^n-1}$
  from a primitive $2^n-1$ root of unity in $GF(2^n)$. Let $A$ be the matrix of the first three rows of $F_{2^n-1}$ and $B$ the matrix formed by adding the columns of $I_3$ to $A$. Then $B$ generates an mds $[2^n+2,3]$ code.
\end{proposition}

\begin{proof} This follows from Lemma \ref{odd} and Corollary \ref{cor}.
\end{proof}

The matrix $A$ in the construction may be formed by taking three rows of
the Fourier matrix, as in Proposition  \ref{cor1}, in arithmetic sequence with difference
$k$ satisfying $\gcd(k,2^n-1)=1$. Then adding in the columns of the
matrix $I_3$ gives
an mds $[2^n+2,3]$ code over $GF(2^n)$. 

By taking the matrix $B$ as the check matrix of a code,  an
$[2^n+2,2^n-1]$ mds code is obtained. %% Alternatively such an mds code
				%% may be formed directly from the
				%% Fourier $F_{2^n} matrix by taking $2^n-1$ rows and adding the columns of $I_3$.   

For even $q$ any $[n,3]$ mds code over $GF(q)$ must satisfy $n\leq q+2$; this is shown in Proposition \ref{jump}. Thus the $[q+2,3]$ mds codes and $[q+2,q-1]$ produced here are best possible length over the field $GF(q)$ for those  dimensions and  even $q$. 

% Work on showing these are best possible will be given in another
% paper. Consideration of the codes here form a basis on which to base a
% proof.
%% \section{Maximality} 
%%  The codes formed here over a prime field are known to have the
%% best possible length for the field. Whether or not this is true in
%% general is considered in another paper by considering the codes
%% developed here as a basis. 

That efficient encoding and decoding
methods are available by adapting the methods of \cite{hurley} is a great advantage.  The complexity of encoding and decoding is $\max\{n\log n, t^2\}$ where $n$ is the length and $t = \floor{\frac{d-1}{2}}$ with distance  $d$.

     \subsection{General $[n,3], [n,n-3]$ } Suppose $G$ is a generator matrix for an
     mds $[n,r]$ code. Then the row-reduced echelon form of $G$ is $(I_r,A)$. Now all the entries in $A$ must be non-zero for
     otherwise an  $r\ti r$ submatrix exists with zero determinant.

     \begin{theorem}\label{2by3} Let $G= (I_r,A)$ be the generator matrix of an $[n,r]$ mds code. Then no $j\ti j$ submatrix of $A$ has $\det = 0$ for $j\leq r$.
     \end{theorem}

In fact: 

\begin{theorem}\label{2by4} Let $G= (I_r,A)$ be the generator matrix of
 an $[n,r]$ code. The code is an mds code  if and only no $j\ti j$ submatrix of $A$ has $\det = 0$ for $j\leq r$.
     \end{theorem}

The proof is not included.

\begin{proposition}\label{jump} Suppose $[n,3]$ is an mds code over $GF(q)$. Then
  $n\leq (q+1)$ when $q$ is odd and $n\leq (q+2)$ when $q$ is even. \end{proposition}
\begin{proof} %% Consider
     %% $GF(q)$ and  $n\geq q+1$.
  Let $\om$ be a primitive element in
     $GF(q)$ and thus $\om$ is a primitive $(q-1)$ root of unity. 
%% Consider a potential  mds $[n,3]$
%%      code $\C$ over $GF(q)$.
     By row operations a generator matrix for the $[n,3]$ mds code has the form
$G=(I_3,A)$ where $A$ is a $3\ti (n-3)$ matrix. If any zero appears in $A$
     then there exists a $3\ti 3$ submatrix with determinant equal to
     zero and then the code would not be an mds code. Hence  
     
     $G= \left(\begin{array}{ccccccc} 1 & 0&0 & \om^{0,1} & \om^{0,2} &\ldots &\om^{0,n-3} \\ 0&1&0 & \om^{1,1}& \om^{1,2} & \ldots & \om^{1,n-3} \\ 0&0&1 & \om^{2,1} & \om^{2,2} &\ldots & \om^{2,n-2} \end{array}\right)$

     where the $\om^{i,j}$ are powers of the primitive element $\om$.
        We are considering $3\ti 3$ submatrices and their determinants so
     we can consider that the first row of the $A$ part of  $G$ consists of $1$'s: 

     $G = \left(\begin{array}{ccccccc} 1&0&0&1&1&\ldots &1 \\  0&1&0 & \om^{1,1}& \om^{1,2} & \ldots & \om^{1,n-3} \\ 0&0&1 & \om^{2,1} & \om^{2,2} &\ldots & \om^{2,n-3} \end{array}\right)$

     If there is a repeat in any of the second or third rows then together with a $(1,1)$ from the first row this gives a $2\ti 2 $ submatrix of $A$ %part
     with determinant $0$. Thus the second and third  rows of the $A$ part of $G$ contain all the elements
     $\{1,\om, \om^2,\ldots, \om^{q-2}\}$ once only.

     Thus $n\leq q+2$ in all cases. 

     No element is repeated in second row and no element is repeated in third row and all powers of $\om$ appear in second and third rows of the $A$ part of $G$.

     Suppose now $q$ is odd and $n=q+2$. Then the sum of the powers of $\om$ in rows $2,3$ is congruent to $\frac{q-1}{2} \mod q$. When subtracting the powers of third row from the powers of the second row above then all the powers must appear or otherwise  a $2\ti 2$ submatrix from second and third rows with determinant $0$ is obtained. Thus the subtraction must result in all the powers appearing and hence the sum of these powers is $\equiv \frac{q-1}{2} \mod q$. But since each of the second and third row sums is $\equiv \frac{q-1}{2} \mod q$, the subtraction results in powers summing to $\equiv 0 \mod q$. Therefore in all cases there exists a $2\ti 2 $ submatrix of $A$ part
     with determinant $0$. Then $G$ has a $3\ti 3$ submatrix with determinant $0$.

     Thus $n\leq q+1$ when $q$ is odd. 

\end{proof}

     %% Consider $q$ even so that the field is $GF(2^m)$. Then $[q+2,3]$ mds
%%      codes have been  constructed in Section \ref{}. 

%% The cases $k=3$ and $k=n-3$ in an mds $[n,k]$ code over a finite field are  given as follows:

In summary then we get the following for dimension $3$. Let $ \C$ be an $[n,3]$ or an $[n,n-3]$ mds code over a  finite field $GF(q)$. Then \begin{enumerate} %% \item $n\leq q+2$ in all  cases;
\item If $q$ is odd then $n\leq (q+1)$. For each such odd $q$, series of examples of $[q+1,3]$ and $[q+1,q-2]$ mds codes with efficient decoding algorithms are constructed using Theorem \ref{main1} of Section \ref{basic}.
\item If $q$ is even then $n\leq q+2$. For each such even $q$, series of examples of $[q+2,3]$ and $[q+2,q-1]$ mds codes with efficient decoding algorithms are constructed using Proposition \ref{cor1}.
  \end{enumerate}
 
%\end{itemize} 
% as follows from the methods in
%      \cite{hurley} as follows: %n=q+2. 

%      $G =  \left(\begin{array}{cccccccc} 1&0&0&1&1&1&\ldots &1 \\ 0&1&0 &
% 		1& \om& \om^2 & \ldots &\om^{q-2} \\ 0&0&1 & 1& \om^2 &
% 		\om^{4} &\ldots & \om^{2(q-2)}
% 		\end{array}\right) $

% $ = 
% 		\left(\begin{array}{ccccccccccc}1&0&0&1&1&1&\ldots & 1& 
% 		 \ldots & \ldots &1 \\  0&1&0 & 
% 		1& \om& \om^2 & \ldots & \om^{q/2-1}& \om^{q/2} & \ldots
% 		 & \om^{q-1}\\ 0&0&1 & 1& \om^2 &
% 		\om^{4} &\ldots & \om^{q-2} & \om & \ldots & \om^{q-3}
% 		      \end{array}\right)$ 
       
% It may be verified that this code is an $[n,3]$ mds code ($n=q+2$) by
% showing that all $3\ti 3$ submatrices have non-zero determinant.  
     
%\section{
     
\end{document}